\documentclass[11pt]{llncs}
\spnewtheorem{observation}{Observation}{\bfseries}{\itshape}
\usepackage{epsfig}
\usepackage{epstopdf}
\usepackage{amstext}
\usepackage{amsmath}

\usepackage[ruled,noend]{algorithm2e}
\spnewtheorem*{examplenonum}{Example}{\it}{\rm}

\begin{document}

\title{Restricted Common Superstring and Restricted Common Supersequence}

\author{
Rapha\"el Clifford\inst{1}, Zvi Gotthilf\inst{2}, Moshe
Lewenstein\inst{2} \and Alexandru Popa\inst{1}}

\institute{
Department of Computer Science, University of Bristol, UK\\
\email{\{clifford,popa\}@cs.bris.ac.uk} \and Department of
Computer Science, Bar-Ilan University, Ramat Gan 52900, Israel
\email{\{gotthiz,moshe\}@cs.biu.ac.il}}

\maketitle

\begin{abstract}


The {\em shortest common superstring} and the {\em shortest common
supersequence} are two well studied problems having a wide range of
applications. In this paper we consider both problems with resource constraints,
denoted as the Restricted Common Superstring (shortly
\textit{RCSstr}) problem and the Restricted Common Supersequence
(shortly \textit{RCSseq}). In the \textit{RCSstr}
(\textit{RCSseq}) problem we are given a set $S$ of $n$ strings,
$s_1$, $s_2$, $\ldots$, $s_n$, and a multiset $t = \{t_1, t_2,
\dots, t_m\}$, and the goal is to find a permutation $\pi : \{1,
\dots, m\} \to \{1, \dots, m\}$ to maximize the number of strings
in $S$ that are substrings (subsequences) of $\pi(t) =
t_{\pi(1)}t_{\pi(2)}...t_{\pi(m)}$ (we call this ordering of the
multiset, $\pi(t)$, a permutation of $t$). We first show that in
its most general setting the \textit{RCSstr} problem is {\em
NP-complete} and hard to approximate within a factor of
$n^{1-\epsilon}$, for any $\epsilon
> 0$, unless P = NP. Afterwards, we present two
separate reductions to show that the \textit{RCSstr} problem
remains NP-Hard even in the case where the elements of $t$ are drawn from a binary alphabet or for the
case where all input strings are of length two. We then present
some approximation results for several variants of the
\textit{RCSstr} problem. In the second part of this paper, we turn
to the \textit{RCSseq} problem, where we present some hardness
results, tight lower bounds and approximation algorithms.
\end{abstract}


\section{Introduction}

\subsection{Motivation}

In AI planning research it is very important to exploit the
interactions between different parts of plans. This was observed
early in the area~\cite{Sacerdoti77,Tate77,Wilkins88}. One very
important type of interaction is the \emph{merging} of different
actions to make the total plan more efficient.

In the general setting we have a set of goals (or tasks) which
have to be accomplished and we want to find the most cost
efficient plan which achieves all the goals. This problem is also
known as the \emph{shortest common superstring} in the case that every
goal has to be done continuously or the \emph{shortest common
supersequence} if we can abandon a task and resume its process
later. In both problems we assume that we have an unlimited set of
resources and we want to achieve all our goals. Of course, in real
life this is never the case: our resources are always limited.

Therefore, a more realistic question is: given a fixed
set of resources, how many goals can be achieved (continuously or
not)?

It seems that most of the applications of the shortest
common superstring and the shortest common supersequence
problem, are more suitable for the case of limited resources. The
main challenge for such applications is to find the best
arrangement that will lead us to accomplish the maximum number of
goals.

As an example, Wilensky~\cite{Wilensky83} gives the scenario where
John is planning to go camping for a
week. He goes to the supermarket to buy a week's worth of
groceries. John has to achieve a set of goals (i.e. to buy food
for meals during the camping weekend) and he is able to merge
some goals (i.e. to buy different products during a single trip to
a supermarket) in order to make the plan more efficient. 

Another application, from the computational biology area, is the
case where only the set of amino acids can be determined and not
their precise ordering. Here we want to know which ordering would
maximize the number of short strings which can be substrings or
subsequences of some ordering of the symbols in a given text.

\subsection{Previous work}

In the shortest common supersequence we are given a set $S$ of $n$
strings, $s_1$,$s_2$,$\ldots$,$s_n$ and we want to find the
shortest string that is a supersequence of every string in $S$.
For arbitrary $n$ the problem is known to be
NP-Hard~\cite{Maier78} even in the case of a binary
alphabet~\cite{RaihaU81}. However for fixed $n$ a dynamic
programming approach takes polynomial time and space. The shortest
common supersequence problem has been studied extensively both
from a theoretical point of
view~\cite{JiangL95,Middendorf93,PEV2002,RubinovT98}, from an
experimental point of view~\cite{BaroneBVM01,Cott05} and from the
perspective of its wide range of applications in data
compression~\cite{Storer88}, query optimization in database
systems~\cite{Sellis98} and text editing~\cite{SankoffK83}.

In the shortest common superstring problem we are given a set $S$
of $n$ strings, $s_1, s_2, \ldots, s_n$ and we want to find the
shortest string that is a superstring of every string in $S$. For
arbitrary $n$ the problem is known to be {\em
NP-Complete}~\cite{GJ79} and APX-hard~\cite{BlumJLTY94}. Even for
the case of binary alphabet Ott~\cite{Ott99}  presented
 lower bounds for the achievable approximation ratio. The best known approximation
ratio so far is 2.5~\cite{KaplanLSS05,Sweedyk99}.

\subsection{Our contributions}

We consider the complexity and the approximability of two problems
which are closely related to the well-known shortest common
superstring and shortest common supersequence problems.

\begin{problem}(Restricted Common Superstring (Supersequence))
The input consists of a set $S = \{s_1, s_2, \dots, s_n\}$ of $n$
strings over an alphabet $\Sigma$ and a multiset $t = \{t_1, t_2,
\dots, t_m\}$ over the same alphabet. The goal is to find an
ordering of the multiset $t$ that maximizes the number of strings in
$S$ that are a substring (subsequence) of the ordered multiset. We denote this
ordering by $\pi(t) = t_{\pi(1)}t_{\pi(2)}...t_{\pi(m)}$ (and we
say that $\pi(t)$ is a permutation of $t$). If all the strings in
$S$ have length at most $\ell$, we refer to the problem as
\textit{RCSstr[$\ell$]} \textit{(RCSseq[$\ell$])}. For simplicity
of presentation, we assume throughout that all the input strings
are distinct and every string $s_i \in S$ is a substring of at
least one permutation $\pi(t)$.
\end{problem}

\begin{example}
Let multiset $t = \{a, a, b, b, c, c\}$ and  set $S =
\{abb,$ $ bbc, cba, aca\}$ be an instance of \textit{RCSstr}
(and also of \textit{RCSstr[3]}). In this example the maximum number of
strings from $S$ that can be a substring of a permutation of $t$
is 3. One such possible permutation is $\pi(t) = acabbc$ which
contains the strings $aca$, $abb$, $bbc$ as substrings.
\end{example}

\begin{example}
Let multiset $t = \{a, a, b, c\}$ and  set $S =
\{ab,bc,cb,ca\}$ be an instance of \textit{RCSseq} and also
\textit{RCSseq[2]}. In this example the maximum number of strings
from $S$ that can be a subsequence of a permutation of $t$ is 3.
One such possible permutation is $\pi(t) = abca$ which contains
the strings $ab$, $bc$, $ca$ as a subsequence.
\end{example}

The paper is organized as follows. In Section~\ref{RCS_hardness}
we study the hardness of the \textit{RCSstr} problem. We show
first that in its most general setting the \textit{RCS} problem is
{\em NP-complete} and hard to approximate within a factor
of less than $n^{1-\epsilon}$, for any $\epsilon > 0$, unless P = NP.
Then, we show that even if all input strings are of length two
(\textit{RCSstr[2]}) and $t$ is a set, i.e. no symbols are repeated, then the \textit{RCSstr} problem is {\em
APX-Hard}. Afterwards, we prove that the \textit{RCSstr} problem
remains NP-Hard even in the case of a binary alphabet.

In Section~\ref{RCS_approx}, we design approximation algorithms
for several restricted variants of the \textit{RCSstr} problem. We
first present a $3/4$ approximation algorithm for the
\textit{RCSstr[2]} problem where $t$ is a set.
Moreover, we give a $1/(\ell(\ell(\ell+1)/2 - 1))$-approximation
algorithm for \textit{RCSstr[$\ell$]}, when $\ell$ is the length of
the longest input string.

The \textit{RCSseq} problem is studied in Section~\ref{PCS}.
In Section~\ref{PCS_hardness} we show that the hardness results
for \textit{RCSstr} hold also for \textit{RCSseq}. Moreover, we
show an approximation lower bound of $1/\ell!$ when $\ell$ is the
length of the longest input string.

In Section~\ref{PCS_approx}, we present approximation
algorithms for two variants of the \textit{RCSseq} problem. The
first is a $(1 + \Omega(1 / \sqrt{\Delta}))/2$ approximation
algorithm for \textit{RCSstr[2]}, where $\Delta$ is the number of
occurrences of the most frequent character in $S$. Then, for
\textit{RCSseq} we show that a selection of an arbitrary
permutation, $\pi(t)$, yields a $1/\ell!$ randomized approximation
algorithm, thus matching the lower bound presented in
Section~\ref{PCS_hardness}.

\section{\textit{RCSstr}}

\subsection{Hardness of the \textit{RCSstr}}
\label{RCS_hardness}

In this section we present hardness results for several variants
of the \textit{RCSstr} problem.

We show here that \textit{RCSstr} problem is {\em NP-complete} and
hard to approximate within a factor better than $n^{1-\epsilon}$,
for any $\epsilon > 0$, unless P = NP. To do so, we present an
approximation-preserving reduction from the classical
\textit{maximum clique} problem.

\begin{definition}(Maximum Clique)
Given an undirected graph $G = (V,E)$ the {\em maximum clique}
problem is to find a vertex set $V' \subseteq V$ of maximum
cardinality, such that for every two vertices in $V'$, there
exists an edge connecting the two.
\end{definition}

The following seminal hardness result will be useful.

\begin{theorem}
\label{th1} \cite{Zuckerman07} The maximum clique problem does not
have an $n^{1-\epsilon}$ approximation, for any $\epsilon > 0$,
unless P = NP.
\end{theorem}

We can now present our main hardness result of the \textit{RCSstr}
problem.

\begin{theorem}
\label{th2} \textit{RCSstr} is {\em NP-complete} and hard to
approximate within a factor of $n^{1-\epsilon}$, for any $\epsilon >
0$, unless P = NP.
\end{theorem}

\begin{proof}

We present an approximation-preserving reduction from the maximum
clique problem to the \textit{RCSstr} problem. Given an undirected
graph $G = (V,E)$, where $V=\{v_1, v_2, \ldots, v_n\}$, we construct an instance
$(S,t)$ of the \textit{RCSstr} problem in the following way.

Set $t$ to be $\{v_1^n, v_2^n, \ldots, v_n^n\}$ and for each vertex $v_i \in V$ define a string $s_i \in S$ as follows. Set $d(v_i)$ to be the ordered sequence of the vertices not adjacent to $v_i$. Set $s_i$ to be $v_i^{n} \cdot d(v_i)$, where $\cdot$ denotes concatenation.


We now prove that the optimal solution of the \textit{RCSstr}
instance $(S,t)$ has size $x$ if and only if the optimal solution
of maximum clique problem on the graph $G$ has size $x$.

Let $\pi$ be a permutation on the multiset $t$ and let $A
\subseteq S$ be all the strings that are substrings of $\pi(t)$.
Denote by $A'$ the set of vertices in $G$ corresponding to the set
of strings $A$. We prove that the vertices in $A'$ form a clique.
Suppose that this is not true and there exist two vertices $v_i,
v_j \in A'$ such that $(v_i,v_j) \notin E$. Note that, in any
common superstring of the strings $s_i$ and $s_j$ either $v_i$ or
$v_j$ must have at least $n+1$ occurrences, since $v_i$ is not
present in the neighbors list of $v_j$ and vice versa. This is a contradiction since the multiset $t$ has only $n$ copies
of each character. Therefore the set of vertices $A'$ forms a
clique.

On the other hand, let $A' = \{v_1, \dots, v_k\} \subseteq V$ be a
clique and let $A = \{s_1, \dots, s_k\} \subseteq S$ be the set of
corresponding strings. We can find a permutation of $t$ which
contains all the strings in $A$ as a substring by concatenating
$s_1, \dots, s_k$ and appending the remaining characters
arbitrarily at the end. No character is used more than $n$ times
since the vertices from $A'$ form a clique and, therefore, $v_i
\notin d(v_j)$ for any $v_i,v_j \in A'$.

Thus, the \textit{RCSstr} problem is {\em NP-complete} and hard to
approximate within a factor  $n^{1-\epsilon}$, for any $\epsilon >
0$, unless P = NP. \qed

\end{proof}

We now show that the \textit{RCSstr[2]} problem is APX-Hard even if
$t$ is a set, i.e. each character in $t$ is unique. To do so, we
present an approximation-preserving reduction from the classical
\textit{Asymmetric maximum TSP} problem with edge weights of $0$
and $1$.

\begin{definition}(Maximum Asymmetric Travelling Salesman
Problem)\\
Given a complete weighted directed graph $G = (V,E)$ the {\em
Maximum Travelling Salesman Problem} is to find a closed tour of
maximum weight visiting all vertices exactly once.
\end{definition}

\begin{theorem}
\label{th3} \cite{EngebretsenK01} For any constant $\epsilon > 0$,
it is NP-Hard to approximate the Maximum Asymmetric Travelling
Salesman with $0$, $1$ edge weights within $320/321 + \epsilon$.
\end{theorem}

The hardness result for the \textit{RCSstr[2]} problem is stated in
the following theorem.

\begin{theorem}
\label{th4} There exists a constant $\beta > 0$, such that the
\textit{RCSstr} problem is NP-Hard to approximate within a factor
of $1-\beta$, even if all the strings in $S$ have length two and
 $t$ is a set.
\end{theorem}

\begin{proof}
We present a gap-preserving reduction from the maximum asymmetric
TSP to the \textit{RCSstr[2]} problem where $t$ is a set.

Given a complete directed graph $G = (V,E)$, with $|V| = n$, $|E|
= n(n-1)/2$ and edge weights of $0$ and $1$, we construct an
instance $(S,t)$ of the \textit{RCSstr[2]} problem in the following
way.

Set $t=V$ and for each arc $(a,b) \in E$ with
weight $1$ set a string $ab$ in $S$. Let $OPT(G)$ be the length of the optimal tour on the
graph $G$ and let $OPT(S,t)$ be the maximum number of strings from
$S$ which can be substrings of a permutation of $t$. In order to
have an inapproximability factor less than $1$, we also assume
that $n > 322$.

We now prove that the reduction presented is a gap-preserving
reduction. Specifically, we prove that:

$$ OPT(G) = n  \Rightarrow  OPT(S,t) = n-1  $$
$$ OPT(G) < (1-\alpha)n \Rightarrow OPT(S,t) < (1-\beta)(n-1) $$

where $\alpha > 0$ and $\beta > 0$ are constants which are defined
later. The permutation $v_1v_2 \dots v_n$ corresponding to a tour
of length $n$ contains $n-1$ strings from $S$ as substrings:
$v_1v_2, v_2v_3, \dots, v_{n-1}v_n$.  Therefore, the first
implication is true.

Suppose now that $OPT(G) < (1-\alpha)n$. Then, $OPT(S,t) <
(1-\alpha)n$, since a permutation of $t$ defines a path in the
graph, which is shorter than a tour. We want to find a constant
$\beta$ such that $(1-\alpha)n \leq (1-\beta)(n-1)$. The following inequality gives the desired.

$$\beta \leq 1 - \frac{1-\alpha}{1-\frac{1}{n}}$$

Therefore, if the maximum ATSP problem does not admit a
$1-\alpha$ approximation, then the \textit{RCSstr[2]} problem (even
in case that $t$ is a set) does not admit a $1-\beta$
approximation (the reader may refer to~\cite{Vaz04} for a more
detailed argument of this claim). From Theorem~\ref{th3}, we know
that is hard to approximate the Maximum Asymmetric Travelling
Salesman with $0$, $1$ edge weights within $320/321 + \epsilon$,
for any $\epsilon > 0$. Therefore, our problem is inapproximable
within 1 - $\beta \geq n(320/321 + \epsilon)/(n-1)$, for any
$\epsilon > 0$.


\qed

\end{proof}

We now show that even over a binary alphabet the \textit{RCSstr}
problem remains NP-Hard.

\begin{theorem}\label{th5}
If $|\Sigma| = 2$, then the \textit{RCSstr} problem is NP-Hard.
\end{theorem}
\begin{proof}
Let $\Sigma = \{0,1\}$. We prove that if we can solve the
\textit{RCSstr} problem on the alphabet $\Sigma$ in polynomial
time, then we can solve in polynomial time the shortest common
superstring problem on the alphabet $\Sigma$.

Consider a shortest common superstring instance $S$, where the
longest string has length $\ell$. It is easy to see that $s_1
\cdot s_2 \cdot \dots \cdot s_n$ is a superstring of all the
strings in $S$. Hence, the solution is no longer than $n\ell$. We
show that $O(n^2\ell^2)$ calls to \textit{RCSstr} are sufficient
to find the shortest common superstring of the given strings.

We name an \textit{RCSstr} instance $(S,t)$ \textit{complete}, if
all the strings of $S$ are substrings of the optimal solution
$\pi(t)$.

Note that there exists a string $x$ with $i$ $0$'s and $j$ $1$'s that is a common
superstring of all the strings in $S$ if and only if the
\textit{RCSstr} instance $(S,0^i1^j)$ is \textit{complete}.
Therefore, we want to find the shortest string $t$ such that the
\textit{RCSstr} instance $(S,t)$ is \textit{complete}. The
shortest common superstring is given by the permutation $\pi(t)$
returned by calling the \textit{RCSstr} on the instance $(S,t)$.
The number of multisets $0^i1^j$ where $i+j \leq n\ell$ is
$O(n^2\ell^2)$. Therefore we can call the \textit{RCSstr} on all
of them and we can find the shortest common superstring on the
given strings in polynomial time (note that this time can be improved somewhat by employing a binary search). The shortest common superstring
problem is NP-Hard and the theorem follows.

%
\qed
\end{proof}

\subsection{Approximating \textit{RCSstr}}
\label{RCS_approx}

In the this section we present approximation algorithms for two
variants of the \textit{RCSstr} problem.

We first present a $3/4$-approximation algorithm for the
\textit{RCSstr[2]} problem where each character of $t$ is unique.
Our algorithm follows immediately from the NP-Hardness reduction
presented in the previous section. Since each character in $t$ is
unique we can construct a complete directed graph $G=(V,E)$, with
$V = \Sigma$ as in the proof of Theorem~\ref{th4}. We then apply
the $3/4$ approximation algorithm for the \textit{Maximum ATSP}
and we obtain a cycle $t_{\pi(1)}, t_{\pi(2)}, \dots,
t_{\pi(n)},t_{\pi(1)}$ of total weight $k$, where
$\pi:\{1,\dots,n\} \to \{1,\dots,n\}$ is a permutation.

If, for some $i < n$, $t_{\pi(i)}t_{\pi(i+1)} \notin S$, we output
$t_{\pi(i+1)}t_{\pi(i+2)} \dots t_{\pi(n-1)}$
$t_{\pi(n)}t_{\pi(1)}t_{\pi(2)} \dots t_{\pi(i)}$, that contains
$k$ strings from $S$ as substrings (and yields an approximation
ratio of $3/4$). Otherwise, we output $t_{\pi(1)} t_{\pi(2)} \dots
$ $t_{\pi(n-1)}t_{\pi(n)}$ that contains exactly $n-1$ strings
from $S$ as substrings, which is optimal.

Here we present a simple $1/(\ell(\ell(\ell+1)/2 -
1))$-approximation algorithm for \textit{RCSstr[$\ell$]}.

The idea is output a concatenation of a maximal collection of
strings from $S$. One can observe that each of the $\ell$
characters of a string in our solution cannot be used by more than
$\ell(\ell+1)/2 - 1$ strings in the optimal solution. Therefore,
the algorithm yields a $1/(\ell(\ell(\ell+1)/2 -
1))$-approximation ratio. Formally, the algorithm is presented
below.

\begin{algorithm}[H]
\caption{A $1/(\ell(\ell(\ell+1)/2 - 1))$ approximation algorithm
for \textit{RCSstr[$\ell$]}} \label{alg_RCSk} Find a maximal subset
$S' = {s'_1, s'_2, \dots, s'_q} \subset S$ of strings under the
following constraint: there exists a permutation $\pi(t)$ of the
multiset such that $s'_1 \cdot s'_2\cdot \dots \cdot s'_q$ is a
prefix of $\pi(t)$. \\

\textit{Output:} $\pi(t)$
\end{algorithm}

\begin{theorem}
Algorithm~\ref{alg_RCSk} is a $1/(\ell(\ell(\ell+1)/2 -
1))$-approximation algorithm for $RCSstr[\ell]$.
\end{theorem}
\begin{proof}

Note that, a single character can be used simultaneously in at
most $\ell(\ell+1)/2 - 1$ strings of the optimal solution. Since
for every $s_i \in S$, $|s_i| \leq \ell$, we can conclude that a
single string in our solution can cause at most
$\ell(\ell(\ell+1)/2 - 1)$ other strings of the optimal solution not to be chosen. Thus, the size of the optimal solution is at
most $q(\ell(\ell(\ell+1)/2 - 1))$ and the approximation ratio
follows. \qed
\end{proof}

One tight example for the above analysis of
Algorithm~\ref{alg_RCSk} is the following: $t =
\{a,b,c,q,q,q,z,z,z,w,w,w,$ $x,x,x\}$, and $S$ = \{$abc$, $qa$,
$az$, $wqa$, $qaz$, $azx$, $qb$, $bz$, $wqb$, $qbz$, $bzx$, $qc$,
$cz$, $wqc$, $qcz$, $czx$\}. If we first select into the maximal
collection the string $abc$, then we cannot add any other string
to our solution. The optimal solution has size $15$ and consists
of all the other strings.

\begin{observation}
Given an $RCSstr[\ell]$ instance, if for every $s_i \in S$, $s_i$ is
not a substring of any other $s_j \in S$, then
Algorithm~\ref{alg_RCSk} is an $\ell^2$-approximation algorithm.
\end{observation}

\begin{proof}
Note that, a single character can be used simultaneously in at
most $\ell$ strings of the optimal solution, thus, a single string
in our solution can stop at most $\ell^2$ other strings of the
optimal solution from being placed.\qed
\end{proof}

One can notice that, in case that all input strings are of length
$\ell$ the above observation must holds.

\section{\textit{RCSseq}}
\label{PCS}

We now turn to the \textit{RCSseq} problem. We first present
hardness results and lower bound for several variants of the
\textit{RCSseq} problem and then we present two approximation
algorithms.

\subsection{Hardness of the \textit{RCSseq} problem}
\label{PCS_hardness}

In the following theorem we show that the hardness result for the
general \textit{RCSstr} holds also to the \textit{RCSseq}.

\begin{theorem}
\label{th6} \textit{RCSseq} is {\em NP-complete} and hard to
approximate within a factor $n^{1-\epsilon}$, for any $\epsilon >
0$, unless P = NP.
\end{theorem}
\begin{proof}
Omitted (similar to the proof of Theorem~\ref{th2}).
\end{proof}

Moreover, we state that even over a binary alphabet the
\textit{RCSseq} problem remains NP-Hard.

\begin{theorem}\label{th7}
If $|\Sigma| = 2$, then the \textit{RCSseq} problem is NP-Hard.
\end{theorem}
\begin{proof}
Omitted (similar to the proof of Theorem~\ref{th3}).
\end{proof}

We now prove that \textit{RCSseq} is APX-Hard even if all the
input strings are of length two and $t$ is a set. To do so, we present an approximation-preserving reduction
from the classical \textit{maximum acyclic subgraph} problem.

\begin{definition}(Maximum Acyclic Subgraph)
Given a directed graph $G = (V,E)$ the {\em maximum acyclic
subgraph} problem is to find a subset $A$ of the arcs such that
$G' = (V,A)$ is acyclic and $A$ has maximum cardinality.
\end{definition}

\begin{theorem}
\label{th8}~\cite{PY:91} The Maximum Acyclic Subgraph problem is
\textit{APX-Complete}.
\end{theorem}

We can now present our hardness result.

\begin{theorem}
\label{th9} \textit{RCSseq} is APX-Hard even if all the strings in
$S$ have length two and $t$ is a set.
\end{theorem}
\begin{proof}
We present an approximation-preserving reduction from the maximum
acyclic subgraph problem. Given a directed graph $G = (V,E)$ we
construct an instance $(S,t)$ of the \textit{RCSseq} problem as follows. Set $t=V$ and for every arc $(a,b) \in E$ we add a string
$ab$ to $S$.

Let $\pi$ be a permutation of the set $t$ and let $A \subseteq S$ be
all the strings that are subsequences of $\pi(t)$. The
corresponding edge set $A$ is an acyclic subgraph of $G$. On the
other hand, let $A \subseteq E$ be an acyclic subgraph. Consider a
topological ordering of $(V,A)$. All strings corresponding to
edges $A$ are subsequences of $\pi(t)$ that corresponds to the
topological ordering.

Note that the optimal solution of the \textit{RCSseq} instance
$(S,t)$ has size $x$ if and only if the optimal solution of
maximum acyclic subgraph problem on the graph $G$ has size $x$.
Thus, the \textit{RCSseq} problem is APX-Hard. \qed

\end{proof}

In~\cite{GMR08} the following result is proven.

\begin{theorem}
The {\em maximum acyclic subgraph} problem is Unique-Games hard to
approximate within a factor better than the trivial $1/2$ achieved
by a random ordering.
\end{theorem}

The maximum acyclic subgraph is a special case of permutation
constraint satisfaction problem (permCSP). A permCSP of arity $k$
is specified by a subset $S$ of permutations on $\{1,2,\dots,k\}$.
An instance of such a permCSP consists of a set of variables $V$
and a collection of constraints each of which is an ordered
$k$-tuple of $V$. The objective is to find a global ordering
$\sigma$ of the variables that maximizes the number of constraint
tuples whose ordering (under $\sigma$) follows a permutation in
$S$. In~\cite{CGM09} Charikar, Guruswami and Manokaran prove the
following result.

\begin{theorem}
For \emph{every} \textit{permCSP} of arity $3$, beating the random
ordering is Unique-Games hard.
\end{theorem}

Our problem corresponds a \textit{permCSP} where $S$ contains only
the identical permutation. Therefore we can conclude the
following.

\begin{theorem}
$RCSseq[2]$ is Unique-Games hard to approximate within a factor better than $1/2$.
\end{theorem}

\begin{theorem}
$RCSseq[3]$ is Unique-Games hard to approximate within a factor better than $1/6$.
\end{theorem}

Currently there is an unpublished result by Charikar, H{\aa}stad
and Guruswami stating that every $k$-ary $permCSP$ is
approximation resistant. This implies that $RCSseq[\ell]$ cannot
have an approximation algorithm better than $1/\ell!$.

\subsection{Approximating \textit{RCSseq}}
\label{PCS_approx}

In the this subsection we present a $(1 + \Omega(1 /
\sqrt{\Delta}))/2 $ approximation algorithm for the \textit{RCSseq[2]}
problem where $\Delta$ is the number of occurrences of the most
frequent character in $S$. We also present a simple randomized
approximation algorithm which achieves an approximation ratio of
$1/\ell!$.

\begin{theorem}\cite{BS:90}
\label{th10} The maximum acyclic subgraph problem is approximable
within $(1 + \Omega(1 / \sqrt{\Delta}))/2$, where $\Delta$ is
the maximum degree of a node in the graph.
\end{theorem}

Given a multiset $t$, let $P'$ be the set of characters that have
a single occurrence in $t$ and let $P$ be $\Sigma \backslash
P'$, where $\Sigma$ is the alphabet of $t$. We define $Q$ to be the following multiset. For every $\sigma
\in P$, if $\sigma$ has $r$ occurrences in $t$, then $\sigma$ has
$r - 2$ occurrences $Q$.

\begin{algorithm}[H]
\caption{A $(1 + \Omega(1 / \sqrt{\Delta}))/2$ approximation
algorithm for \textit{RCSseq2}} \label{alg1}
\begin{enumerate}
\item Given a multiset $t$, construct a graph $G = (V,E)$ such
that:\\ $v_i \in V$ iff $v_i \in P'$ and $(a,b) \in E$ iff $a,b
\in P'$ and $ab \in S$.

\item Apply the $(1 + \Omega(1 / \sqrt{\Delta}))/2$ approximation
algorithm for the maximum acyclic subgraph to the graph G. Denote
the output subgraph by $G'(V,E')$.


\item Let $F'$ be a topological order of the vertices of $G'$.\\
Let $F$ and $F''$ be an arbitrary ordering of $P$ and $Q$
respectively.

\item Output $F\cdot F'\cdot F\cdot F''$.
\end{enumerate}
\end{algorithm}

Figure~\ref{fig2} is an example of Algorithm~\ref{alg1}. In the
first stage we construct a graph according to the first two steps,
note that $P=\{e\}$, $P' = \{a,b,c,d\}$ and $Q=\emptyset$. Then we
present an acyclic directed subgraph and we output $F \cdot
F'\cdot F \cdot F''$, where $F=e$ and $F'=cadb$.

\begin{figure}
\begin{center}
\epsfig{file=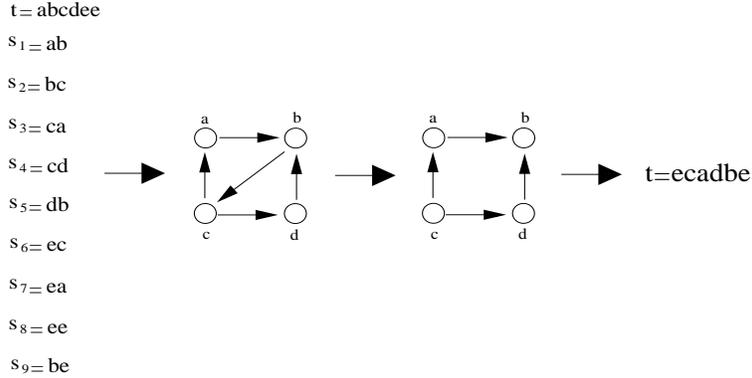,height=2.0in,width=4.0in}\\
\caption{\sf Algorithm~\ref{alg1} example.}\label{fig2}
\end{center}
\end{figure}

\begin{theorem}
\label{th11} Algorithm~\ref{alg1} is a $(1 + \Omega(1 /
\sqrt{\Delta}))/2$ approximation algorithm for the \textit{RCSseq[2]}
problem, where $\Delta$ is the maximum number of occurrences of a
character in the set $S$.
\end{theorem}

\begin{proof}
Given a string $ab \in S$. If $a \in P$ or $b \in P$ (or both),
then $ab$ is always a subsequence of $F \cdot F'\cdot F$.
Otherwise, if both $a$ and $b$ appear only once in $t$, then $ab$
is a subsequence of $F \cdot F'\cdot F$ if only if the edge
$(a,b)$ is selected in the arc set of the maximum acyclic
subgraph. Since the maximum acyclic subgraph problem has an
approximation ratio of $(1 + \Omega(1 / \sqrt{\Delta}))/2$, the
same approximation ratio holds for \textit{RCSseq2} problem.  \qed
\end{proof}

We now deal with \textit{RCSseq[$\ell$]} instances. We show that
selecting an arbitrary permutation $\pi(t)$ achieves an expected
approximation ratio of $\frac {1}{\ell!}$.

We define by $P(s_i,\pi(t))$ the probability that a string $s_i
\in S$ is a subsequence of a permutation $\pi(t)$.

Note that, $P(s_i,\pi(t)) \geq \frac{{|t| \choose
\ell}(|t|-\ell)!}{|t|!} = \frac{1}{\ell!}$. Therefore, the
expected number of strings from $S$ to be subsequences of an
arbitrary permutation $\pi(t) \geq \frac{|S|}{\ell!}$. Thus,
selecting an arbitrary permutation $\pi(t)$ achieves an expected
approximation ratio of at least $\frac{|S|}{|S|\ell!}= \frac
{1}{\ell!}$.



\bibliography{bib-latest}

\begin{thebibliography}{10}

\bibitem{BaroneBVM01}
Paolo Barone, Paola Bonizzoni, Gianluca~Della Vedova, and Giancarlo Mauri.
\newblock An approximation algorithm for the shortest common supersequence
  problem: an experimental analysis.
\newblock In {\em SAC}, pages 56--60, 2001.

\bibitem{BS:90}
Bonnie Berger and Peter~W. Shor.
\newblock Approximation algorithms for the maximum acyclic subgraph problem.
\newblock In {\em SODA}, pages 236--243, 1990.

\bibitem{BlumJLTY94}
Avrim Blum, Tao Jiang, Ming Li, John Tromp, and Mihalis Yannakakis.
\newblock Linear approximation of shortest superstrings.
\newblock {\em J. ACM}, 41(4):630--647, 1994.

\bibitem{CGM09}
Moses Charikar, Venkatesan Guruswami, and Rajsekar Manokaran.
\newblock Every permutation csp of arity 3 is approximation resistant.
\newblock In {\em IEEE Conference on Computational Complexity}, pages 62--73,
  2009.

\bibitem{Cott05}
Carlos Cotta.
\newblock Memetic algorithms with partial lamarckism for the shortest common
  supersequence problem.
\newblock In {\em IWINAC (2)}, pages 84--91, 2005.

\bibitem{EngebretsenK01}
Lars Engebretsen and Marek Karpinski.
\newblock Approximation hardness of tsp with bounded metrics.
\newblock In {\em ICALP}, pages 201--212, 2001.

\bibitem{GJ79}
M.~R. Garey and D.~S. Johnson.
\newblock {\em Computers and intractability. A guide to the theory of
  NP-completeness}.
\newblock W. H. Freeman, 1979.

\bibitem{GMR08}
Venkatesan Guruswami, Rajsekar Manokaran, and Prasad Raghavendra.
\newblock Beating the random ordering is hard: Inapproximability of maximum
  acyclic subgraph.
\newblock In {\em FOCS}, pages 573--582, 2008.

\bibitem{JiangL95}
Tao Jiang and Ming Li.
\newblock On the approximation of shortest common supersequences and longest
  common subsequences.
\newblock {\em SIAM J. Comput.}, 24(5):1122--1139, 1995.

\bibitem{KaplanLSS05}
Haim Kaplan, Moshe Lewenstein, Nira Shafrir, and Maxim Sviridenko.
\newblock Approximation algorithms for asymmetric tsp by decomposing directed
  regular multigraphs.
\newblock {\em J. ACM}, 52(4):602--626, 2005.

\bibitem{Maier78}
David Maier.
\newblock The complexity of some problems on subsequences and supersequences.
\newblock {\em J. ACM}, 25(2):322--336, 1978.

\bibitem{Middendorf93}
Martin Middendorf.
\newblock The shortest common nonsubsequence problem is np-complete.
\newblock {\em Theor. Comput. Sci.}, 108(2):365--369, 1993.

\bibitem{Ott99}
Sascha Ott.
\newblock Lower bounds for approximating shortest superstrings over an alphabet
  of size 2.
\newblock In {\em WG}, pages 55--64, 1999.

\bibitem{PY:91}
Christos~H. Papadimitriou and Mihalis Yannakakis.
\newblock Optimization, approximation, and complexity classes.
\newblock {\em J. Comput. Syst. Sci.}, 43(3):425--440, 1991.

\bibitem{PEV2002}
P.~A. Pevzner.
\newblock Multiple alignment, communication cost, and graph matching.
\newblock {\em SIAM Journal of Applied Mathematics}, 52(6):1763--1779, December
  1992.

\bibitem{RaihaU81}
Kari-Jouko R{\"a}ih{\"a} and Esko Ukkonen.
\newblock The shortest common supersequence problem over binary alphabet is
  np-complete.
\newblock {\em Theor. Comput. Sci.}, 16:187--198, 1981.

\bibitem{RubinovT98}
Anatoly~R. Rubinov and Vadim~G. Timkovsky.
\newblock String noninclusion optimization problems.
\newblock {\em SIAM J. Discrete Math.}, 11(3):456--467, 1998.

\bibitem{Sacerdoti77}
Earl~D. Sacerdoti.
\newblock A structure for plans and behavior.
\newblock {\em American Elsevier}, 1977.

\bibitem{SankoffK83}
David Sankoff and Joseph Kruskal.
\newblock {\em Time Warps, String Edits, and Macromolecules: The Theory and
  Practice of Sequence Comparison}.
\newblock CSLI Publications, 1983.

\bibitem{Sellis98}
Timos~K. Sellis.
\newblock Multiple-query optimization.
\newblock {\em ACM Trans. Database Syst.}, 13(1):23--52, 1988.

\bibitem{Storer88}
James~A. Storer.
\newblock {\em Data Compression: Methods and Theory}.
\newblock Computer Science Press, 1988.

\bibitem{Sweedyk99}
Z.~Sweedyk.
\newblock A 2$\frac{1}{2}$-approximation algorithm for shortest superstring.
\newblock {\em SIAM J. Comput.}, 29(3):954--986, 1999.

\bibitem{Tate77}
Austin Tate.
\newblock Generating project networks.
\newblock In {\em IJCAI}, pages 888--893, 1977.

\bibitem{Vaz04}
Vijay~V. Vazirani.
\newblock {\em Approximation Algorithms}.
\newblock Springer, 2004.

\bibitem{Wilensky83}
R.~Wilensky.
\newblock Planning and understanding.
\newblock {\em Addison Wesley}, 1983.

\bibitem{Wilkins88}
David~E. Wilkins.
\newblock Practical planning: Extending the classical ai planning paradigm.
\newblock {\em Morgan Kaufmann, CA}, 1988.

\bibitem{Zuckerman07}
David Zuckerman.
\newblock Linear degree extractors and the inapproximability of max clique and
  chromatic number.
\newblock {\em Theory of Computing}, 3(1):103--128, 2007.

\end{thebibliography}
\bibliographystyle{plain}

\end{document}